%% file: main.tex
\documentclass[runningheads]{llncs}
\usepackage{amsmath}
\usepackage[T1]{fontenc}
\usepackage{graphicx}
\usepackage{amssymb}
\usepackage[ruled,vlined,linesnumbered]{algorithm2e}
\usepackage{booktabs}
\usepackage{listings}
\usepackage{tikz}
\usepackage{microtype}
\usepackage{bbding}
\usetikzlibrary{arrows.meta,patterns,positioning,matrix,decorations.pathreplacing}

\DeclareMathOperator{\CFL}{CFL}
\DeclareMathOperator{\ICFL}{ICFL}

\DeclareMathOperator{\V}{L}

\DeclareMathOperator{\NE}{next}
\DeclareMathOperator{\Bo}{border}
\DeclareMathOperator{\lce}{lce}
\newcommand{\quotes}[1]{``#1''}
\newcommand{\proc}[1]{\textnormal{\textsc{#1}}}
\title{The Inverse Lyndon Array}
\titlerunning{The Inverse Lyndon Array}
\author{Clelia De Felice\inst{1} \and
Pietro Negri\inst{1} \and
Manuel Sica\inst{1} \Envelope \and
Rocco Zaccagnino\inst{1} \and
Rosalba Zizza\inst{1}\Envelope}
\authorrunning{C. De Felice, P. Negri, M. Sica, R. Zaccagnino, and R. Zizza}

\institute{Dipartimento di Informatica, University of Salerno, Fisciano (SA), Italy\\
\email{cdefelice@unisa.it, p.negri137@gmail.com, masica@unisa.it, rzaccagnino@unisa.it, rzizza@unisa.it}}

\begin{document}

\maketitle

\begin{abstract}
The Lyndon array stores, at each position of a word, the length of the longest Lyndon factor starting at that position and plays an important role in combinatorics on words, for example, in the construction of fundamental data structures such as the suffix array.
In this paper, we introduce the \emph{Inverse Lyndon array}, the analogous structure for inverse Lyndon words, namely words that are lexicographically greater than all their proper nonempty suffixes. Unlike standard Lyndon words, inverse Lyndon words may have non-trivial borders, which introduces a genuine theoretical difficulty. We show that the Inverse Lyndon array can be characterized in terms of the next greater suffix array together with a border-correction term, and we prove that this correction coincides with a longest common extension (LCE) value. Building on this characterization, we adapt the nearest-suffix framework underlying Ellert's linear-time construction of the Lyndon array to the inverse setting, obtaining an $O(n)$-time algorithm for general ordered alphabets. 
Finally, we show that the Inverse Lyndon array can also be used to reconstruct the canonical inverse Lyndon factorization in linear time. 
\keywords{Inverse Lyndon words \and ICFL factorization \and Lyndon array \and Combinatorial algorithms}
\end{abstract}

\section{Introduction}

Lyndon words, widely studied in combinatorics on words, are primitive words that are lexicographically minimal among their rotations~\cite{lyndon-words}. Equivalently, a nonempty word $x$ is a Lyndon word if $x \prec s$ for every nonempty proper suffix $s$ of $x$, where $\prec$ is the lexicographic order induced by a total order on the alphabet~$\Sigma$.

The \emph{Lyndon array} $\lambda[1..n]$ of a word $x[1..n]$ gives, at each position $i$, the length of the longest Lyndon factor of $x$ starting at~$i$. 
Bannai et al.~\cite{runs-theorem} showed that $\lambda$ is sufficient to compute all maximal repetitions in linear time and so efficient construction algorithms for $\lambda$ have received considerable attention. In \cite{LOUZA} an algorithm
for computing $\lambda$ as a byproduct of the inversion of the bijective Burrows-Wheeler Transform is presented.
Franek et al.~\cite{algorithms-lyndon-array-original,algorithms-lyndon-array-revisited} related $\lambda$ to a \textit{next-smaller-value} computation on the inverse suffix array, while Ellert~\cite{lyndon-simple} later gave the first direct linear-time algorithm on general ordered alphabets, based on nearest smaller suffix array and a longest common extension (LCE) acceleration mechanism inspired by Manacher's algorithm~\cite{Manacher1975ANL}. Subsequent work considered space-efficient variants, sub-linear algorithms, and links with the Lyndon forest~\cite{badkobeh_et_al:LIPIcs.CPM.2022.13,sublinear-time,space-efficient}.

A closely related notion is that of \emph{inverse Lyndon words}. A nonempty word $x$ is an inverse Lyndon  word if $s \prec x$ for every nonempty proper suffix~$s$ of $x$~\cite{inverse-lyndon}. This notion
was introduced in connection with a variant of Chen-Fox-Lyndon ($\CFL$) factorization, the \emph{canonical inverse Lyndon factorization} ($\ICFL$), 
that may yield more balanced factors and therefore well suited to bioinformatics applications~\cite{lyndon-bio-example}. A border property of $\ICFL$ was established in~\cite{lata2020,tcs2021}, while a simpler right-to-left construction algorithm was given in~\cite{inverse-lyndon-2}. The inverse setting is not a mere dualization of the standard one~\cite{inverse-lyndon}: standard Lyndon words are unbordered, whereas inverse Lyndon words may have non-trivial borders.
 As a consequence, the classical identity $\lambda[i]=\mathit{next}[i]-i$ shown in \cite{lyndon-simple}, which relates 
the longest Lyndon word starting in $i$, with 
its nearest smaller suffix,
no longer survives unchanged, when we want to 
define the \emph{Inverse Lyndon array}~$\lambda^{-1}[1..n]$, where $\lambda^{-1}[i]$ is the length of the longest inverse Lyndon factor starting at position~$i$.
The main difficulty is not merely to reverse the inequalities in Ellert's framework, but to characterize how borders interact with nearest greater suffixes (NGS) while preserving the amortized linearity of the underlying LCE machinery. The \quotes{border correction} is crucial: unless it is identified with an LCE value on the NGS, a separate border computation would be required, preventing a direct linear-time reduction. Our result thus connects non-crossing LCE queries on general ordered alphabets~\cite{noncrossing-lce,ellert-runs-general} with the border-based view of inverse Lyndon words and $\ICFL$~\cite{inverse-lyndon-2}. Specifically, we show that the required border correction can be recovered directly from the LCE on the NGS. Moreover, the Inverse Lyndon array contains enough information to recover the canonical inverse Lyndon factorization in linear time, establishing a direct algorithmic connection between $\lambda^{-1}$ and $\ICFL$.

We briefly summarize our main contributions: the
definition of the 
Inverse Lyndon array $\lambda^{-1}$ and the
exact characterization of $\lambda^{-1}$ via the \emph{nearest greater suffix} (NGS) array and a border correction term (Lemma~\ref{lem:equiv-inv});
the identification of this correction with an LCE value on the NGS, avoiding explicit border computation (Lemma~\ref{lem:border-lce});
the adaptation of the non-crossing property, chain iteration, and LCE acceleration from the NSS/PSS setting to the NGS/PGS setting, yielding a linear-time algorithm for constructing $\lambda^{-1}$ over general ordered alphabets, together with an experimental evaluation in the Appendix;
a linear-time reduction from $\lambda^{-1}$ to $\ICFL$, showing that $\lambda^{-1}$ suffices to efficiently reconstruct $\ICFL$.

\section{Preliminaries}\label{sec:prelim}

\subsection{Words and Lexicographic Orders}
As in \cite{lyndon-simple}, in order to denote the integers $\{i, i+1, \ldots, j\}$,
we use the notations $[i,j]=[i, j+1)=(i-1,j]=(i-1,j+1)$.
A \emph{word} $x = x[1]x[2]\cdots x[n]$ is a sequence of symbols from a totally ordered alphabet $(\Sigma, <)$. We write $|x| = n$ for its length, $\varepsilon$ for the empty word, and $\Sigma^*$ and $\Sigma^+$ for the sets of all words and all nonempty words, respectively, on $\Sigma$. For $1 \le i \le j \le n$, the factor $x[i..j]$ is the sequence $x[i]x[i{+}1]\cdots x[j]$; if $i > j$, it equals $\varepsilon$ and if $x[i..j] \neq x$, it is a proper factor. 
A factor $x[1..j]$ is a \emph{prefix} of $x$
and $x[j..n]$ is a \emph{suffix} of $x$. 
The word $x_i = x[i..n]$ denotes the suffix of
$x$ starting at position~$i$.
We recall that, given a nonempty word $x$, a {\em border} of $x$ is a word which is both a proper prefix and a suffix of $x$. The longest proper prefix of $x$ which is a suffix of $x$ is also called {\em the border} of $x$.

The \emph{lexicographic order} $\prec$ on $\Sigma^*$ is defined as follows: $x \prec y$ (or $y \succ x$) 
if and only if either $y = xw$ for some nonempty word $w$, or $x = urv$ and $y = usw$ with $r < s$, $r,s \in \Sigma$.
For two nonempty words $x,y$, we write $x \ll y$ if
$x \prec y$ and $x$ is not a proper prefix of $y$
\cite{Bannai15}. We also write $y \succ x$ if $x \prec y$.
Basic properties of the lexicographic order are recalled below.

\begin{lemma} \label{proplexord}
For $x,y \in \Sigma^*$, the following properties hold.
\begin{itemize}
\item[(1)]
$x \prec y$ if and only if $zx \prec zy$,
for every word $z$.
\item[(2)]
If $x \ll y$, then $xu \ll yv$
for all words $u,v$.
\item[(3)]
If $x \prec y \prec xz$ for a word $z$,
then $y = xy'$ for some word $y'$ such that $y' \prec z$.
\end{itemize}
\end{lemma}

The \emph{Longest Common Extension} of two suffixes 
$x_i, x_j$ of $x$, denoted by $\lce(i,j)$, is
the length of the longest common prefix of the suffixes, formally, $\lce(i,j)=|x_i|$ if $i = j$; otherwise
$\lce(i,j) = |u|$ if $x_i = uay, x_j = ubz, a \not = b.$

Following \cite{lyndon-simple}, we assume that every word starts and ends with special symbols $\#, \$ \not \in \Sigma$, named {\em sentinels}, i.e.,  $x = \#\,x(1..n)\,\$$. In addition,  $\#$ and $\$$ satisfy $\# < \$ < a$ for all $a \in \Sigma$ in the standard setting, and $\# > \$ > a$ for all $a \in \Sigma$ in the inverse setting. In the standard case, placing $\#$ strictly below every alphabet symbol ensures that $x_1$ is the lexicographic minimum among all suffixes of the word. In the inverse case, the reversed order makes $x_1$ the lexicographic maximum. The sentinels allows to avoid that a suffix is a prefix of another one, simplifying 
algorithms.

\begin{definition}[\cite{PIERREDUVAL1983363}]
A nonempty word $x$ is a \emph{Lyndon word} if $x \prec s$ for every nonempty proper suffix $s$.
\end{definition}

We denote by $\V $ the set of Lyndon words relative to the lexicographic order
on $\Sigma^*$.
The following is from \cite{PIERREDUVAL1983363,reu}.
A {\it sesquipower} of a word $x$ is a word $x^kp$ where
$p$ is a proper prefix of $x$ and $k \geq 0$. A sesquipower is called {\it nontrivial} if
$k \geq 1$.
Denote by $S$ the set of nontrivial sesquipowers of Lyndon words.
Clearly any element of $S$ has the form $(uv)^ku$, where $uv$ is a Lyndon word,
$u \in \Sigma^*$, $v \in \Sigma^+$, $k \geq 1$ and this representation is unique (see \cite{reu} for a proof).
It is also clear that $\V \subseteq S$.
The following result characterizes, for a given $z \in S$ and a letter $b$, whether $zb$ is still
in $S$ or not and, in the first case, whether $zb$ is a Lyndon word or not \cite{PIERREDUVAL1983363}.
Duval's algorithm for computing $\CFL$ is based on this result.

\begin{theorem} \label{FundamentalPreneck}
Let $z = (uav')^ku \in S$, where $uav'$ is a Lyndon word,
$u, v' \in \Sigma^*$, $a \in \Sigma$, $k \geq 1$. For any $b \in \Sigma$, the
word $zb$ is still in $S$ if and only if $a \leq b$. Moreover $zb$ is a Lyndon word if and only if
$a < b$.
\end{theorem}

\subsection{The Lyndon Array and the NSS/PSS Framework}\label{sec:nss}

A factor $x[i..j]$ is the \emph{longest Lyndon factor} of 
$x$ starting at~$i$ if it is a Lyndon word and either $j = n$ or $x[i..j{+}1]$ is not a Lyndon word.

\begin{definition}[Lyndon Array
\cite{algorithms-lyndon-array-original}]\label{def:lyn-arr}
The \emph{Lyndon array} $\lambda_x[1..n]$ of a word $x[1..n]$ is defined by
\[
\lambda_x[i] = \max\{m \in [1,n-i+1] \mid x[i..i+m-1] \mbox{ is a Lyndon word}\}.
\]
\end{definition}

In the sequel, if the context does not make it ambiguous, we 
do not use the subscript $x$
and $x = \#\,x(1..n)\,\$$, with $\# > \$ > a$ for all $a \in \Sigma$.
As a template for our inverse construction, we recall Ellert's definitions and results for computing $\lambda$ in $O(n)$ time~\cite{lyndon-simple}.

\begin{definition}[NSS and PSS edges]\label{def:nss-pss}
For a word $x[1..n]$, define
\[
\begin{array}{rcl}
\mathit{next}[i] & = & \min\{j \in (i,n] \mid x_j \prec x_i\}, \\
\mathit{prev}[i] & = & \max\{j \in [1,i) \mid x_j \prec x_i\},
\end{array}
\]
with $\mathit{next}[1] = \mathit{next}[n] = n{+}1$, $\mathit{prev}[1] = 0$, and $\mathit{prev}[n] = 1$.
\end{definition}

The sets above are nonempty, as they contain $\$$ and $\#$ respectively.
We provide an alternative proof of the lemma below.

\begin{lemma}\label{lem:equiv-std}
For all $i \in [1,n]$, $\mathit{next}[i] = i + \lambda[i]$.
\end{lemma}
\begin{proof}
Let $\NE[i] = h$. Therefore, the longest Lyndon factor of $x$ starting
from $i$ is $x[i .. h-1]$.
Consequently
$\lambda[i] = |x[i .. h-1]| = h - i = \NE[i] - i$.
\end{proof}


The following properties enable the linear-time construction of~\cite{lyndon-simple}.

\begin{lemma}[Non-Crossing]\label{lem:no-crossing}
Let $l_1 < r_1$ and $l_2 < r_2$ be index pairs, each connected by an NSS or PSS edge. Then it is impossible to have $l_1 < l_2 < r_1 < r_2$.
\end{lemma}

Given $r \in [1,n]$ and an integer $e\geq 0$,
a {\em chain of previous smaller suffixes}
is recursively defined as follows: 
$prev^0[r]=r$, $prev^{e+1}[r]=prev^e[prev[r]]$.
We denote $\ell=prev^*[r]$, if there is an integer $e \geq 0$ s.t. $\ell=prev^e[r]$.

Let $x=\#x(1..n)\$$ be a word with previous and
next smaller suffix arrays $prev$ and $next$.

\begin{lemma}[Chain Iteration]\label{lem:chain}
For any $\ell, r \in [1,n]$:
\begin{enumerate}
\item[(i)] $\mathit{prev}[r] = \mathit{prev}^*[r{-}1]$.
\item[(ii)] $\mathit{next}[\ell] = r$ if and only if $\ell = \mathit{prev}^*[r{-}1]$ and $\ell > \mathit{prev}[r]$.
\end{enumerate}
\end{lemma}

\begin{lemma}[LCE Acceleration]\label{lem:lce-accel}
Let $k \in (1,n)$, $\ell = \mathit{prev}[k]$ and $r = \mathit{next}[k]$. Then:
\begin{enumerate}
\item[(i)] if $\mathrm{lce}(\ell,k) = \mathrm{lce}(k,r)$, then $\mathrm{lce}(\ell,r) \ge \mathrm{lce}(k,r)$ and either $prev[r]= \ell$ or $next[\ell]=r$.
\item[(ii)] if $\mathrm{lce}(\ell,k) < \mathrm{lce}(k,r)$, then $\mathrm{lce}(\ell,r) = \mathrm{lce}(\ell,k)$ and $\mathit{prev}[r] = \ell$.
\item[(iii)] if $\mathrm{lce}(\ell,k) > \mathrm{lce}(k,r)$, then $\mathrm{lce}(\ell,r) = \mathrm{lce}(k,r)$ and $\mathit{next}[\ell] = r$.
\end{enumerate}
\end{lemma}

Together with the SMART-LCE procedure of~\cite{lyndon-simple}, which computes all required LCE values in amortized $O(n)$ time via a global rightmost-inspected-position variable, these lemmas yield an algorithm that computes $\lambda$ in $O(n)$ time and $O(n)$ space on general ordered alphabets.

On the other hand, the computation of the Lyndon array is closely related to Duval's algorithm for calculating 
$\CFL(x)$ since this algorithm computes the longest Lyndon prefix of a word $x$.
In turn, Duval's algorithm is based on Theorem \ref{FundamentalPreneck} and computes the longest nontrivial sesquipower
of $x$ which is a prefix of $x$ (see \cite{PIERREDUVAL1983363,reu}).
These observations are formalized in the following proposition, where $\Bo(x)$ denotes the length of the border of $x$.

\begin{proposition} \label{main}
Let $x[i \ldots i+m-1]$ be the longest nontrivial sesquipower of $x$ starting
from $i$. Then the following facts hold true.
\begin{enumerate}
\item
$x[i .. i+m-1] = (uav')^ku$, where $uav'$ is a Lyndon word,
$u, v' \in \Sigma^*$, $a \in \Sigma$, $k \geq 1$ and $x[i+m] = b$, 
where $a > b$;
\item
$uav'$ is the longest Lyndon factor of $x$ starting
from $i$;
\item
$|uav'| = |x[i .. i+m-1]| - \Bo(x[i .. i+m-1])$.
\end{enumerate}
\end{proposition}

\section{The Inverse Lyndon array}\label{sec:inverse}

We now introduce the Inverse Lyndon array. 
We recall the definition of {inverse Lyndon words} and for their
interesting properties we refer to
\cite{DLT22,inverse-lyndon-2,inverse-lyndon,lata2020,tcs2021}.

\begin{definition}\label{def:inv-lyn}
A word $w \in \Sigma^+$ is an \emph{inverse Lyndon word} if $s \prec w$ 
for every nonempty proper suffix~$s$ of~$w$.
\end{definition}

Unlike standard Lyndon words, inverse Lyndon words may have non-trivial borders; for instance, $\mathit{dabda}$ has border $\mathit{da}$. We now introduce the analogous notion of longest inverse Lyndon factor.

\begin{definition}[Longest Inverse Lyndon Factor]\label{def:inv-max}
A factor $x[i..j]$ is a \emph{longest inverse Lyndon factor} starting at~$i$ if it is an inverse Lyndon word and either $j = n$ or $x[i..j{+}1]$ is not an inverse Lyndon word.
\end{definition}

\begin{definition} \label{ILO}
Let $(\Sigma, <)$ be a totally ordered alphabet.
The {\rm inverse} $<_{in}$ of $<$ is defined
by
$b <_{in} a \Leftrightarrow a < b, \quad \forall a, b \in \Sigma$.
The {\rm inverse lexicographic order} (or simply {\rm inverse order})
on $\Sigma^*$, denoted $\prec_{in}$, is the lexicographic order
on $\Sigma^*$ defined by $(\Sigma, <_{in})$.
\end{definition}

Following \cite{antiLyndon}, a Lyndon word relative to the inverse lexicographic order
on $\Sigma^*$ will be named an {\it anti-Lyndon word}.
Proposition \ref{InvChar}, proved in \cite{inverse-lyndon}
states a fundamental relation between inverse Lyndon words and anti-Lyndon words.

\begin{proposition} \label{InvChar}
Let $x \in \Sigma^+$. Then
$x$ is an inverse Lyndon word if and only if
$x$ is a nontrivial sesquipower of an anti-Lyndon word.
\end{proposition}

\subsection{The Inverse Lyndon array and NGS/PGS Edges}

In the following, $x = \#\,x(1..n)\,\$$, with $\# > \$ > a$ for all $a \in \Sigma$.

\begin{definition}[Inverse Lyndon array]\label{def:inv-arr}
The \emph{Inverse Lyndon array} $\lambda_x^{-1}[1..n]$ of a word $x[1..n]$ is defined by
\[
\lambda_x^{-1}[i] = \max\{m \in [1,n-i+1] \mid x[i..i+m-1] \mbox{ is an inverse Lyndon word}\}.
\]
\end{definition}

In the sequel, we do not use the subscript $x$.

\begin{definition}[NGS and PGS Arrays]\label{def:ngs-pgs}
For a word $x[1..n]$, we define:
\[
\begin{array}{rcl}
\mathit{next}^{-1}[i] & = & \min\{j \in (i,n] \mid x_j \succ x_i\}, \\
\mathit{prev}^{-1}[i] & = & \max\{j \in [1,i) \mid x_j \succ x_i\},
\end{array}
\]
with $\mathit{next}^{-1}[1] = \mathit{next}^{-1}[n] = n{+}1$, $\mathit{prev}^{-1}[1] = 0$, and $\mathit{prev}^{-1}[n] = 1$.
\end{definition}

Thus, in the inverse setting, the relevant edges connect nearest \emph{greater} suffixes instead of 
nearest smaller suffixes.

\begin{example}\label{ex:inv-arr}
Let $x = \#\mathit{aababbaa}\$$, with $\# > \$ > b > a$. The Inverse Lyndon array and the associated NGS/PGS arrays are:
\begin{center}
\small
$\begin{array}{l|cccccccccc}
i & 1 & 2 & 3 & 4 & 5 & 6 & 7 & 8 & 9 & 10 \\
x & \# & a & a & b & a & b & b & a & a & \$ \\
\hline
\lambda^{-1}        & 10 & 2 & 1 & 3 & 1 & 4 & 3 & 2 & 1 & 1 \\
\mathit{next}^{-1}  & 11 & 3 & 4 & 6 & 6 & 10 & 10 & 9 & 10 & 11 \\
\mathit{prev}^{-1}  & 0  & 1 & 1 & 1 & 4 & 1 & 6 & 7 & 7 & 1 \\
\mathit{border}       & 0  & 1 & 0 & 1 & 0 & 0 & 0 & 1 & 0 & 0 \\
\end{array}$
\end{center}
At $i=4$, the longest inverse Lyndon factor is $x[4..6] = \mathit{bab}$, so $\lambda^{-1}[4] = 3$. Its border has length $1$, hence $\mathit{next}^{-1}[4] = 6$ and $\mathit{border}[4] = 1$ (as we will see, in agreement with Corollary~\ref{cor:recover}). At $i=6$, the longest inverse Lyndon factor is $x[6..9] = \mathit{bbaa}$, so $\lambda^{-1}[6] = 4$. Since this factor is unbordered, $\mathit{next}^{-1}[6] = 10$ and $\mathit{border}[6] = 0$. Figure~\ref{fig:ngs-arcs} shows the corresponding inverse nearest-suffix structure on the same word. As done in \cite{lyndon-simple},
$\mathit{next}^{-1}$ and
$\mathit{prev}^{-1}$ arrays are represented by directed edges 
connecting the positions $i$ to $j$,
whenever $\mathit{prev}^{-1}[i]=j$ (PGS edge) or $\mathit{next}^{-1}[i]=j$
(NGS edge). 
\end{example}

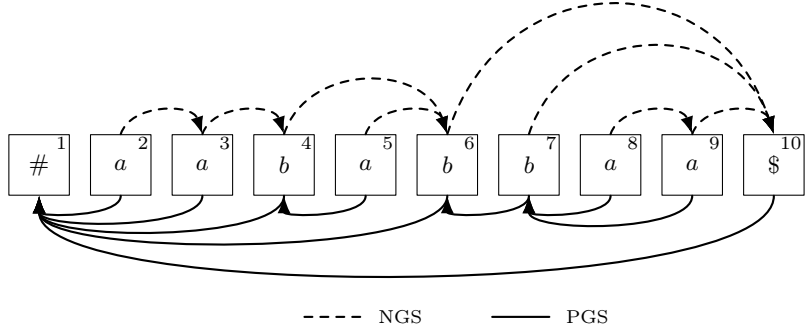
\begin{figure}[tbp]
\centering
\begin{tikzpicture}[x=1.08cm,y=0.82cm,>=Latex,line cap=round,line join=round]

  \foreach \i/\c in {1/\#,2/a,3/a,4/b,5/a,6/b,7/b,8/a,9/a,10/\$} {
    \node[draw,minimum width=8mm,minimum height=8mm,inner sep=0pt] (c\i) at (\i,0) {\footnotesize $\c$};
    \node[anchor=north east,font=\scriptsize,inner sep=0.5pt] at ([xshift=-0.5pt,yshift=-0.5pt]c\i.north east) {\i};
  }

  \draw[->,thick,dashed] (c2.north) to[out=72,in=108,looseness=1.02] (c3.north);
  \draw[->,thick,dashed] (c3.north) to[out=72,in=108,looseness=1.02] (c4.north);
  \draw[->,thick,dashed] (c4.north) to[out=72,in=108,looseness=1.18] (c6.north);
  \draw[->,thick,dashed] (c5.north) to[out=72,in=108,looseness=1.02] (c6.north);
  \draw[->,thick,dashed] (c6.north) to[out=72,in=108,looseness=1.42] (c10.north);
  \draw[->,thick,dashed] (c7.north) to[out=72,in=108,looseness=1.28] (c10.north);
  \draw[->,thick,dashed] (c8.north) to[out=72,in=108,looseness=1.02] (c9.north);
  \draw[->,thick,dashed] (c9.north) to[out=72,in=108,looseness=1.02] (c10.north);

  \draw[->,thick] (c2.south) .. controls +(0,-0.35) and +(0,-0.35) .. (c1.south);
  \draw[->,thick] (c3.south) .. controls +(0,-0.55) and +(0,-0.55) .. (c1.south);
  \draw[->,thick] (c4.south) .. controls +(0,-0.75) and +(0,-0.75) .. (c1.south);
  \draw[->,thick] (c5.south) .. controls +(0,-0.35) and +(0,-0.35) .. (c4.south);
  \draw[->,thick] (c6.south) .. controls +(0,-1.00) and +(0,-1.00) .. (c1.south);
  \draw[->,thick] (c7.south) .. controls +(0,-0.35) and +(0,-0.35) .. (c6.south);
  \draw[->,thick] (c8.south) .. controls +(0,-0.35) and +(0,-0.35) .. (c7.south);
  \draw[->,thick] (c9.south) .. controls +(0,-0.60) and +(0,-0.60) .. (c7.south);
  \draw[->,thick] (c10.south) .. controls +(0,-1.70) and +(0,-1.70) .. (c1.south);

  \draw[thick,dashed] (4.25,-2.45) -- (4.95,-2.45);
  \node[anchor=west,font=\scriptsize] at (5.05,-2.45) {NGS};
  \draw[thick] (6.55,-2.45) -- (7.25,-2.45);
  \node[anchor=west,font=\scriptsize] at (7.35,-2.45) {PGS};

\end{tikzpicture}
\caption{NGS and PGS edges for $x=\#aababbaa\$$. Dashed arcs denote next greater suffix edges, while solid arcs denote previous greater suffix edges.}
\label{fig:ngs-arcs}
\end{figure}

Let us denote by
$x_{\lambda^{-1}}(i)$ the word starting at position
$i$ of length $\lambda^{-1}[i]$, 
i.e., the longest inverse Lyndon
word starting at position $i$.
We now prove the relation between
the Inverse Lyndon array and the NGS edge:
$x_{\lambda^{-1}}(i)$ ends
at the starting position of the NGS of $x_i$ minus
the length of the border of $x_{\lambda^{-1}}(i)$,
denoted 
by $\mathrm{border}(x_{\lambda^{-1}}(i))$.

\begin{lemma}[NGS and $\lambda^{-1}$]\label{lem:equiv-inv}
For all $i \in [1,n]$,
$\mathit{next}^{-1}[i]
= i + \lambda^{-1}[i] - \mathrm{border}(x_{\lambda^{-1}}(i))$.
\end{lemma}

\begin{proof} 
Let $\lambda^{A}$ be the Lyndon array of a word $x = x[1 \ldots n]$
with respect to the inverse lexicographical order.
By Lemma \ref{lem:equiv-std} and Proposition \ref{main}, we have
$$\NE^{-1}[i] = i + \lambda^{A}[i] = i + \lambda^{-1}[i] - \Bo(x_{\lambda^{-1}}(i)).$$
\end{proof}

\begin{lemma}[Border via LCE]\label{lem:border-lce}
Let $z = x_{\lambda^{-1}}(i)$ be the longest inverse Lyndon factor of $x$ starting at $i$
and let $h = \NE^{-1}[i]$. Then
$\lce(i,h) = \Bo(z)$.
\end{lemma}

\begin{proof}
Let $\widehat{\Sigma}=\Sigma\cup\{\#,\$\}$.
By Proposition \ref{main},
we have $z = (uav')^ku$, $h = i + |uav'|$,
$x_h = (uav')^{k - 1} uby$, with $b > a$ and $y \in \widehat{\Sigma}^*$.
Hence
$\lce(i,h) = |(uav')^{k - 1} u| = \Bo(z).$
\end{proof}

\begin{corollary}\label{cor:recover}
For every $i$,
$
\lambda^{-1}[i] = \mathit{next}^{-1}[i] - i + \mathrm{\lce}(i,\mathit{next}^{-1}[i]).
$
\end{corollary}

Two additional examples illustrating the border correction, including the case of an overlapping border, are provided in Appendix~\ref{sec:further-border-examples}.

\begin{remark}[Why the border term is structural]\label{rem:structural}
The correction term in Corollary~\ref{cor:recover} is not a post-processing detail. Without Lemma~\ref{lem:border-lce}, the inverse array would require explicit border information for each longest inverse Lyndon factor, introducing an additional computational layer not present in the standard case. The point of Lemma~\ref{lem:border-lce} is precisely that the border is already encoded by the same nearest-suffix LCE information computed by the algorithm (see
[Algorithm 1 in \cite{lyndon-simple}] for detailed description of the computing LCE array using two arrays, $\mathit{nlce}$ and $\mathit{plce}$,
such that $\mathit{plce}[i]=\mathrm{\lce}(\mathit{prev}[i],i)$ and
$\mathit{nlce}[i]=\mathrm{\lce}(i,\mathit{next}[i])$ i.e., the border array in our framework). Thus, computing $\mathit{next}^{-1}$ together with the associated $\mathit{\lce}$ array suffices to recover $\lambda^{-1}$. 
\end{remark}

\subsection{Combinatorial Properties of NGS/PGS Edges}

We now show that the main structural lemmas from the NSS/PSS setting remain valid in the inverse case. These lemmas also fit into the broader context of non-crossing LCE-based techniques on general ordered alphabets. Here, however, the non-crossing structure is induced specifically by NGS/PGS edges in the inverse Lyndon setting. The proofs of Lemmas \ref{lem:inv-no-crossing}
and \ref{lem:inv-chain} are similar to the ones in
\cite{lyndon-simple} and are reported in the Appendix~\ref{sec:technicalproofs}
for the sake of completeness.

\begin{lemma}[Non-Crossing of NGS/PGS Edges]\label{lem:inv-no-crossing}
Let $l_1 < r_1$ and $l_2 < r_2$ be index pairs, each connected by an NGS or PGS edge. Then it is impossible to have $l_1 < l_2 < r_1 < r_2$.
\end{lemma}

\begin{lemma}[Chain Iteration for NGS/PGS]\label{lem:inv-chain}
For any $\ell, r \in [1,n]$:
\begin{enumerate}
\item $\mathit{prev}^{-1}[r] = (\mathit{prev}^{-1})^*[r{-}1]$;
\item $\mathit{next}^{-1}[\ell] = r$ if and only if $\ell = (\mathit{prev}^{-1})^*[r{-}1]$ and $\ell > \mathit{prev}^{-1}[r]$.
\end{enumerate}
\end{lemma}

\begin{lemma}[LCE Acceleration for NGS/PGS]\label{lem:inv-lce-accel}
Let $\ell = \mathit{prev}^{-1}[k]$ and $r = \mathit{next}^{-1}[k]$. Then:
\begin{enumerate}
\item if $\mathrm{\lce}(\ell,k) = \mathrm{\lce}(k,r)$, then $\mathrm{\lce}(\ell,r) \ge \mathrm{lce}(k,r)$, and either $\mathit{prev}^{-1}[r] = \ell$ or $\mathit{next}^{-1}[\ell] = r$;
\item if $\mathrm{\lce}(\ell,k) < \mathrm{\lce}(k,r)$, then $\mathrm{\lce}(\ell,r) = \mathrm{\lce}(\ell,k)$ and $\mathit{prev}^{-1}[r] = \ell$;
\item if $\mathrm{\lce}(\ell,k) > \mathrm{\lce}(k,r)$, then $\mathrm{\lce}(\ell,r) = \mathrm{\lce}(k,r)$ and $\mathit{next}^{-1}[\ell] = r$.
\end{enumerate}
\end{lemma}

\begin{proof}
For the first case, since $\ell = \mathit{prev}^{-1}[k]$ and $r = \mathit{next}^{-1}[k]$, we have $x_{\ell} \succ x_k \succ x_m$ and $x_r \succ x_k \succ x_m$ for all $m \in (\ell,r) \setminus \{k\}$. Since $\mathrm{\lce}(\ell,k) = \mathrm{\lce}(k,r)$, it follows that $\mathrm{\lce}(\ell,r) \ge \mathrm{\lce}(k,r)$. If $x_{\ell} \succ x_r$, then 
$x_{\ell} \succ x_r \succ x_m$ for all such $m$, so $\mathit{prev}^{-1}[r] = \ell$. By symmetry, if $x_r \succ x_{\ell}$, then $\mathit{next}^{-1}[\ell] = r$.

For the second case, let $c = \mathrm{\lce}(\ell,k)$. Since $x_{\ell} \succ x_k$, we have $x[\ell+c] > x[k+c]$. Because $\mathrm{\lce}(\ell,k) < \mathrm{\lce}(k,r)$, the suffix $x_r$ shares a prefix of length at least $c+1$ with $x_k$, hence $x[r+c] = x[k+c]$. Thus $x_{\ell}$ and $x_r$ agree for $c$ positions and differ at the next one with 
$x[\ell+c] > x[r+c]$, so $x_{\ell} \succ x_r$ and $\mathrm{\lce}(\ell,r) = c$. Since $x_r \succ x_m$ for all $m \in (\ell,r)$, it follows that $\mathit{prev}^{-1}[r] = \ell$.

The third case is symmetric, yielding $x_r \succ x_{\ell}$, $\mathrm{\lce}({\ell},r) = \mathrm{\lce}(k,r)$, and $\mathit{next}^{-1}[\ell] = r$.
\end{proof}

\subsection{The LCE-NGS Algorithm}

Lemmas~\ref{lem:inv-chain} and~\ref{lem:inv-lce-accel} yield Algorithm~\ref{alg:lce-ngs}, the inverse counterpart of the LCE-NSS algorithm~\cite{lyndon-simple}. The only difference is the while-loop comparison: we test $x[\ell+m] < x[r+m]$ to find the first lexicographically greater suffix to the right.

\begin{algorithm}[h]
\caption{LCE-NGS for the Inverse Lyndon array}
\label{alg:lce-ngs}
\DontPrintSemicolon
\SetAlgoLined
\SetInd{0.4em}{0.9em}

\KwIn{word $x[1..n]$ in inverse sentinel mode}
\KwOut{$\mathit{next}^{-1}$, $\mathit{prev}^{-1}$, $\mathit{nlce}$, $\mathit{plce}$, $\lambda^{-1}$}

$\mathit{prev}^{-1}[1] \gets 0$\;
$\mathit{next}^{-1}[n] \gets n+1$\;
$\mathit{plce}[1] \gets 0$\;

\For{$r \gets 2$ \KwTo $n$}{
  $\ell \gets r-1$\;
  $m \gets \proc{Smart-Lce}(\ell,r,0)$\;

  \While{$x[\ell+m] < x[r+m]$}{
    $\mathit{next}^{-1}[\ell] \gets r$\;
    $\mathit{nlce}[\ell] \gets m$\;

    \uIf{$m = \mathit{plce}[\ell]$}{
      $m \gets \proc{Smart-Lce}(\mathit{prev}^{-1}[\ell],r,m)$\;
    }
    \uElseIf{$m > \mathit{plce}[\ell]$}{
      $m \gets \mathit{plce}[\ell]$\;
    }

    $\ell \gets \mathit{prev}^{-1}[\ell]$\;
  }

  $\mathit{prev}^{-1}[r] \gets \ell$\;
  $\mathit{plce}[r] \gets m$\;
}

\For{$i \gets 1$ \KwTo $n$}{
  $\lambda^{-1}[i] \gets \mathit{next}^{-1}[i] - i + \mathit{nlce}[i]$\;
}
\end{algorithm}

The procedure \proc{Smart-Lce} is the amortized primitive of~\cite{lyndon-simple}, here applied only to NGS/PGS candidates. Its proof relies on two invariants: $\mathit{rhs}$ is the rightmost inspected position, and each fixed edge stores its exact LCE in $\mathit{nlce}$ or $\mathit{plce}$. Thus, if $r+m<\mathit{rhs}$, Lemma~\ref{lem:inv-lce-accel} retrieves the answer in $O(1)$ time; otherwise, \proc{Smart-Lce} scans further, updates $\mathit{rhs}$, and stores the new LCE. Only this latter case performs new character comparisons.

\begin{theorem}\label{thm:linear}
Algorithm~\ref{alg:lce-ngs} computes the Inverse Lyndon array $\lambda^{-1}[1..n]$ in $O(n)$ time and $O(n)$ space on general ordered alphabets.
\end{theorem}

\begin{proof}
Correctness follows from Lemmas~\ref{lem:equiv-inv}, \ref{lem:border-lce}, \ref{lem:inv-no-crossing}, \ref{lem:inv-chain}, and~\ref{lem:inv-lce-accel}. The additional point to justify is that the border correction does not create extra rescanning.

By Lemma~\ref{lem:border-lce}, whenever $\mathit{next}^{-1}[\ell]=r$ is fixed, the same value $\mathit{nlce}[\ell]=\mathrm{\lce}(\ell,r)$ is both the edge LCE needed by the algorithm and the border correction needed later for $\lambda^{-1}[\ell]$. Therefore the algorithm never computes borders explicitly and never performs a separate pass for border information.

Moreover, if \proc{Smart-Lce} is called with $r+m<\mathit{rhs}$, then the query lies entirely inside a region already certified by a previous explicit scan. In that case, Lemma~\ref{lem:inv-lce-accel} and the stored $\mathit{nlce}$ or $\mathit{plce}$ values suffice to answer the query without rescanning characters inside that window. Thus borders cannot trigger hidden rescanning there. Every explicit character comparison strictly advances the global frontier $\mathit{rhs}$, which is monotone and never exceeds $n$. Hence the total number of explicit comparisons is $O(n)$.

The outer loop performs $n-1$ iterations. Each index receives its final $\mathit{next}^{-1}$ value once, each stored LCE value is written once, and the final recovery formula
$
\lambda^{-1}[i] = \mathit{next}^{-1}[i]-i+\mathit{nlce}[i]
$
is applied once per position. Therefore, the total running time is $O(n)$ and the space usage is $O(n)$.
\end{proof}

\section{From the Inverse Lyndon array to $\ICFL$}
\label{sec:icfl-recovery}

In this section, we show that the Inverse Lyndon array can be used to reconstruct the \emph{canonical inverse Lyndon factorization} ($\ICFL$) of a word in linear time.

\begin{definition}\cite{inverse-lyndon}
An inverse Lyndon factorization of a word $w \in \Sigma^+$ is a sequence
$(m_1, \ldots,m_k)$ of inverse Lyndon words such that $m_1 \cdots m_k = w$ and $m_i \prec m_{i+1}$, 
$m_i \neq m_{i+1}$,
$1 \leq i \leq k-1$.
\end{definition}

Given a word, its inverse Lyndon factorization
is not unique and so in \cite{inverse-lyndon}
the $\ICFL$ is also recursively defined. However,
here we report a simpler characterization,
provided in \cite{inverse-lyndon-2}.

\begin{proposition}
Let $w \in \Sigma^+$. The \emph{canonical inverse Lyndon factorization} of $w$,
denoted by $\ICFL(w)$, is the unique inverse Lyndon factorization
$\ICFL(w) = (m_1,\dots,m_k)$
such that each $m_\ell$ is an inverse Lyndon word and the factorization satisfies
the {\em border property}, namely no nonempty border of $m_\ell$ is a prefix of
$m_{\ell+1}$, for every $\ell=1,\dots,k-1$.
\end{proposition}

The key property is the following: for every nonempty word $w$, the last factor of $\ICFL(w)$ coincides with the longest suffix of $w$ that is an inverse Lyndon word \cite{inverse-lyndon-2}. Equivalently, if $w = w'm$, where $m$ is the longest suffix of $w$ that is an inverse Lyndon word, then $\ICFL(w) = (\ICFL(w'),m)$.
The proofs of the properties introduced in this section, and related examples, can be found in Appendix \ref{sec:icfl-recovery-app}.
The lemma below relates the Inverse Lyndon array of a word to the Inverse Lyndon array of the same word with sentinels.
\begin{lemma}
\label{lem:restrict-lambda-inv}
Let $w = w[1..n] \in \Sigma^*$ and let $x = \#\,w\,\$$, with $\# > \$ > a$ for every $a \in \Sigma$. Then, for every $i \in [1,n]$, we have $\lambda^{-1}_w[i] = \min\bigl(\lambda^{-1}_x[i+1],\, n-i+1\bigr)$.
\end{lemma}
 
Let us denote $R_i = i + \lambda^{-1}_w[i] - 1$ 
for $i=1,\dots,n$. Then, by definition, $w[i..R_i]=w_{\lambda^{-1}}(i)$ 
is the longest factor of $w$ starting at position $i$ that is an inverse Lyndon word. 

Lemma \ref{lem:inv-cover} highlights the connection between $w[i..j]$ and $R_i$.

\begin{lemma}
\label{lem:inv-cover}
For every pair of integers $i$ and $j$ such that $1 \le i \le j \le n$,
\[
w[i..j] \mbox{ is an inverse Lyndon word}
\quad\Longleftrightarrow\quad
R_i \ge j.
\]
\end{lemma}

Now, we define $C_j = \{\, i \in [1,j] \mid R_i \ge j \,\}$ and $L_j = \min C_j$, for each $j \in [1,n]$. 
Observe that $j \in C_j$, since $\lambda^{-1}_w[j] \ge 1$. Thus, $C_j$ is the set of starting positions of inverse Lyndon factors covering position $j$, while $L_j$ is the leftmost such starting position.

Lemma \ref{lem:longest-suffix-prefix} and Lemma \ref{lem:L-monotone} show some interesting properties about $L_j$.

\begin{lemma}
\label{lem:longest-suffix-prefix}
For every $j \in [1,n]$, the factor $w[L_j..j]$ is the longest suffix of the prefix $w[1..j]$ that is an inverse Lyndon word.
\end{lemma}

\begin{lemma}
\label{lem:L-monotone}
The sequence $(L_j)_{j=1}^n$ is nondecreasing.
\end{lemma}

Now we can show that $R_j$ and $L_j$ can be computed in linear time and space. Specifically, Algorithm~\ref{alg:compute-L} takes as input $\lambda^{-1}_w$ and computes the arrays $R$ and $L$, where $R[j]=R_j$ and $L[j]=L_j$ for every $j=1,\dots,n$.

\begin{algorithm}[h]
\caption{\proc{Compute-L}$(\lambda^{-1}[1..n])$}
\label{alg:compute-L}
\DontPrintSemicolon
\SetAlgoLined
\SetInd{0.4em}{0.9em}

\KwIn{$\lambda^{-1}[1..n]$}
\KwOut{$R[1..n]$, $L[1..n]$}

\For{$i \gets 1$ \KwTo $n$}{
  $R[i] \gets i + \lambda^{-1}[i] - 1$\;
}
$p \gets 1$\;
\For{$j \gets 1$ \KwTo $n$}{
  \While{$R[p] < j$}{
    $p \gets p+1$\;
  }
  $L[j] \gets p$\;
}
\Return{$R,L$}\;
\end{algorithm}

\begin{proposition}
\label{prop:compute-L}
Algorithm \proc{Compute-L} correctly computes the arrays $R[1..n]$ and $L[1..n]$ in $O(n)$ time and $O(n)$ space.
\end{proposition}

In addition $R_j$ and $L_j$ can be used to reconstruct $\ICFL(w)$ for a word $w$. 
Taking $L$ as input, Algorithm~\ref{alg:recover-icfl} reconstructs $\ICFL(w)$ from right to left.

\begin{algorithm}[h]
\caption{\proc{Recover-ICFL}$(w[1..n],\lambda^{-1}[1..n])$}
\label{alg:recover-icfl}
\DontPrintSemicolon
\SetAlgoLined
\SetInd{0.4em}{0.9em}

\KwIn{word $w[1..n]$, Inverse Lyndon array $\lambda^{-1}[1..n]$}
\KwOut{$\ICFL(w)$}

compute $(R,L)$ using \proc{Compute-L}$(\lambda^{-1})$\;
$j \gets n$\;
$F \gets$ empty list\;
\While{$j \ge 1$}{
  $s \gets L[j]$\;
  add the factor $w[s..j]$ to the front of $F$\;
  $j \gets s-1$\;
}
\Return{$F$}\;
\end{algorithm}

In the following, we establish the results required to prove the correctness of Algorithm~\ref{alg:recover-icfl}, which computes $\ICFL(w)$ from \(\lambda_w^{-1}\). More precisely, Lemma~\ref{lem:single-step-icfl} shows that, at each step, Algorithm~\ref{alg:recover-icfl} correctly determines the last factor of \(\ICFL(u)\), where \(u\) is the current prefix of \(w\).

\begin{lemma}
\label{lem:single-step-icfl}
Let $u = w[1..j]$ be the prefix currently considered by \proc{Recover-ICFL}, and let $m = u[L[j]..j]$. Then $m$ is the last factor of $\ICFL(u)$.
\end{lemma}

\begin{proof}
By Lemma~\ref{lem:longest-suffix-prefix}, $u[L[j]..j]$ is the longest suffix of $u$ that is inverse Lyndon word. By Proposition~7.7 in~\cite{inverse-lyndon-2}, if $\ICFL(u)=(m_1,\dots,m_k)$ then $m_k$ is precisely the longest suffix of $u$ that is an inverse Lyndon word. Thus $m = u[L[j]..j] = m_k$. Finally, by Remark~7.8~\cite{inverse-lyndon-2}, $\ICFL(u)=(\ICFL(u'),m)$.
\end{proof}

\begin{theorem}
\label{thm:recover-icfl-correct}
Algorithm \proc{Recover-ICFL} correctly computes $\ICFL(w)$.
\end{theorem}
\begin{proof}
The algorithm works from right to left on \(w\). At each iteration, let $u = w[1..j]$ be the current prefix. It computes $s=L[j]$ and extracts the factor $m = u[s..j]$.
By Lemma~\ref{lem:single-step-icfl}, \(m\) is the last factor of \(\ICFL(u)\). Therefore, writing $u = u'm$, with $u' = w[1..s-1]$, we have $\ICFL(u) = (\ICFL(u'), m)$.
The algorithm then continues on \(u'\), and the same argument applies at each subsequent iteration. Hence, it successively extracts the factors of \(\ICFL(w)\) from right to left. Since each extracted factor is added to the front of the list \(F\), the final output preserves the correct left-to-right order. Therefore, the algorithm returns \(\ICFL(w)\).
\end{proof}

\begin{theorem}
\label{thm:recover-icfl-linear}
Algorithm \proc{Recover-ICFL} computes $\ICFL(w)$ in $O(n)$ time and $O(n)$ space.
\end{theorem}
\begin{proof}
By Proposition~\ref{prop:compute-L}, the computation of the arrays \(R\) and \(L\) requires time \(O(n)\) and space \(O(n)\). In the reconstruction phase, the current value of \(j\) is updated as $j \gets L[j]-1$. Since $1 \le L[j] \le j$, the value of \(j\) decreases strictly at each iteration. Hence the total number of steps is at most \(n\).
At each step, the algorithm performs a constant number of accesses to the arrays and determines exactly one new factor. If the factors are stored by means of their endpoints, the cost of each step is \(O(1)\). Even if the substrings are explicitly materialized, the sum of the lengths of all output factors is exactly \(n\), since they form a factorization of \(w\).
Therefore, the reconstruction phase requires total time \(O(n)\). Adding the preprocessing cost, we obtain total time \(O(n)\) and total space \(O(n)\).
\end{proof}

\section{Conclusions and Future Perspectives}\label{sec:suffix}

We introduced the Inverse Lyndon array, characterizing it via the NGS array and a border correction term. Then we provide a linear-time construction over general ordered alphabets by adapting the NSS/PSS framework to NGS/PGS. The correction is captured by LCE values on NGS edges, preserving the same amortized linear-time behavior, also confirmed experimentally. Finally, we showed how to reconstruct $\ICFL$ from the Inverse Lyndon array in linear time.

A natural perspective is to face the
problem of suffix sorting by using the Inverse Lyndon array. It is well known that the Lyndon array supports suffix sorting through a compatibility property~\cite{baier:LIPIcs.CPM.2016.23,accelerate-suffix-sorting}. 
Even if a compatibility property still holds for
$\ICFL$ \cite{tcs2021}, a strategy for suffix sorting based on inverse Lyndon words 
has not been proposed.
A future work is to combine the results presented in this paper 
with other properties of $\ICFL$ to define efficient approaches to suffix sorting,
a direction that also aligns with the broader formal language perspective discussed in~\cite{DLT22}.

\newpage
\appendix
\input{appendix}

\begin{credits}
\subsubsection{\ackname}
{\it Clelia De Felice}: from MUR 2022YRB97K, PINC, Pangenome INformatiCs: from Theory to Applications.
{\it Manuel Sica and Rosalba Zizza}: from INdAM - GNCS Project CUP\_E53C25002010001.

\subsubsection{\discintname}
The authors declare that they have no competing interests relevant to this article.
\end{credits}

\begingroup
\sloppy
\hbadness=2000
\emergencystretch=1.5em
\bibliography{references}
\endgroup

\end{document}

%% file: appendix.tex
\section{Further Examples of the Border Correction}
\label{sec:further-border-examples}

The following two examples illustrate the border-correction term in Corollary~\ref{cor:recover}, first when the border overlaps with itself and then when it does not.

\begin{example}[Overlapping border]\label{ex:overlapping-border}
Let $x=\#\mathit{bababab}\$$, with $\# > \$ > b > a$. The Inverse Lyndon array and the associated NGS/PGS arrays are:
\begin{center}
\small
$\begin{array}{l|ccccccccc}
i & 1 & 2 & 3 & 4 & 5 & 6 & 7 & 8 & 9 \\
x & \# & b & a & b & a & b & a & b & \$ \\
\hline
\lambda^{-1}       & 9 & 7 & 1 & 5 & 1 & 3 & 1 & 1 & 1 \\
\mathit{next}^{-1} & 10 & 4 & 4 & 6 & 6 & 8 & 8 & 9 & 10 \\
\mathit{prev}^{-1} & 0 & 1 & 2 & 1 & 4 & 1 & 6 & 1 & 1 \\
\mathit{border}    & 0 & 5 & 0 & 3 & 0 & 1 & 0 & 0 & 0 \\
\end{array}$
\end{center}
At $i=2$, the longest inverse Lyndon subword is $x[2..8]=\mathit{bababab}$, so $\lambda^{-1}[2]=7$. Its longest border is $\mathit{babab}$, of length $5$, and the two occurrences of this border overlap. We have $\mathit{next}^{-1}[2]=4$ and $\mathrm{lce}(2,4)=5$, hence
\[
\lambda^{-1}[2]
=\mathit{next}^{-1}[2]-2+\mathrm{lce}(2,\mathit{next}^{-1}[2])
=4-2+5=7.
\]
Thus the border correction also captures an overlapping border.
\end{example}

\begin{example}[Non-overlapping border]\label{ex:nonoverlapping-border}
Let $x=\#\mathit{baaba}\$$, with $\# > \$ > b > a$. The Inverse Lyndon array and the associated NGS/PGS arrays are:
\begin{center}
\small
$\begin{array}{l|ccccccc}
i & 1 & 2 & 3 & 4 & 5 & 6 & 7 \\
x & \# & b & a & a & b & a & \$ \\
\hline
\lambda^{-1}       & 7 & 5 & 2 & 1 & 2 & 1 & 1 \\
\mathit{next}^{-1} & 8 & 5 & 4 & 5 & 7 & 7 & 8 \\
\mathit{prev}^{-1} & 0 & 1 & 2 & 2 & 1 & 5 & 1 \\
\mathit{border}    & 0 & 2 & 1 & 0 & 0 & 0 & 0 \\
\end{array}$
\end{center}
At $i=2$, the longest inverse Lyndon subword is $x[2..6]=\mathit{baaba}$, so $\lambda^{-1}[2]=5$. Its longest border is $\mathit{ba}$, of length $2$, and the two occurrences do not overlap. We have $\mathit{next}^{-1}[2]=5$ and $\mathrm{lce}(2,5)=2$, hence
\[
\lambda^{-1}[2]
=\mathit{next}^{-1}[2]-2+\mathrm{lce}(2,\mathit{next}^{-1}[2])
=5-2+2=5.
\]
\end{example}

\section{Technical Proofs}
\label{sec:technicalproofs}

\noindent
\textbf{Proof of Lemma \ref{lem:inv-no-crossing}}

\begin{proof}
Assume for contradiction that $l_1 < l_2 < r_1 < r_2$. Since $l_2 \in (l_1,r_1)$ and $l_1,r_1$ are connected by an NGS or PGS edge, every suffix starting in $(l_1,r_1)$ is lexicographically smaller than both endpoints, hence $x_{l_2} \prec x_{r_1}$. On the other hand, since $r_1 \in (l_2,r_2)$ and $l_2,r_2$ are connected, we get $x_{l_2} \succ x_{r_1}$. This contradiction proves the claim.
\end{proof}

\medskip
\noindent
\textbf{Proof of Lemma \ref{lem:inv-chain}}

\begin{proof}
For the first statement, suppose $\mathit{prev}^{-1}[r] \neq (\mathit{prev}^{-1})^*[r{-}1]$. Then there exists $r' = (\mathit{prev}^{-1})^*[r{-}1]$ such that $\mathit{prev}^{-1}[r] \in (\mathit{prev}^{-1}[r'],r')$, giving
\[
\mathit{prev}^{-1}[r'] < \mathit{prev}^{-1}[r] < r' < r,
\]
which contradicts Lemma~\ref{lem:inv-no-crossing}.

For the second statement, assume first that $\mathit{next}^{-1}[\ell] = r$. Then all suffixes starting in $(\ell,r)$ are smaller than $x_r$, so $\mathit{prev}^{-1}[r] \notin [\ell,r)$ and therefore $\mathit{prev}^{-1}[r] < \ell$. If $\ell \neq (\mathit{prev}^{-1})^*[r{-}1]$, then there exists $r' = (\mathit{prev}^{-1})^*[r{-}1]$ such that
\[
\mathit{prev}^{-1}[r'] < \ell < r' < r,
\]
again contradicting Lemma~\ref{lem:inv-no-crossing}.

Conversely, assume that $\mathit{prev}^{-1}[r] < \ell$ and $\ell = (\mathit{prev}^{-1})^*[r{-}1]$. Since $\ell \in (\mathit{prev}^{-1}[r],r)$, we have $x_{\ell} \prec x_r$. Moreover, the chain property implies $x_k \prec x_{\ell}$ for every $k \in (\ell,r)$. Therefore $r$ is the first position to the right of $\ell$ with a greater suffix, namely $\mathit{next}^{-1}[\ell] = r$.
\end{proof}

\section{From the Inverse Lyndon array to ICFL}
\label{sec:icfl-recovery-app}

\medskip
\noindent
\textbf{Proof of Lemma \ref{lem:restrict-lambda-inv}}

\begin{proof}
By definition, the longest factor of $w$ starting at position $i$ that is an inverse Lyndon word coincides with the longest factor of $x$ starting at position $i+1$ that is an inverse Lyndon word. Such a factor cannot extend beyond the end of $w$, hence its length is at most $n-i+1$. Therefore,
\[
\lambda^{-1}_w[i] = \min\bigl(\lambda^{-1}_x[i+1],\, n-i+1\bigr).
\]
\end{proof}

\medskip
\noindent
The following prefix-closure property is fundamental~\cite{inverse-lyndon}.

\begin{lemma}\label{lem:prefix-closed}
Every nonempty prefix of an inverse Lyndon word is an inverse Lyndon word.
\end{lemma}

\textbf{Proof of Lemma \ref{lem:inv-cover}}
\begin{proof}
Assume that $R_i \ge j$. Then $w[i..j]$ is a non-empty prefix of $w[i..R_i]$. Since inverse Lyndon words are prefix-closed by Lemma~\ref{lem:prefix-closed}, it follows that $w[i..j]$ is also an inverse Lyndon word.

Conversely, suppose that $w[i..j]$ is an inverse Lyndon word. By definition of $\lambda^{-1}_w[i]$, we have
\[
\lambda^{-1}_w[i] \ge |w[i..j]| = j-i+1.
\]
Hence,
\[
R_i = i + \lambda^{-1}_w[i] - 1 \ge i + (j-i+1) - 1 = j.
\]
\end{proof}

\medskip
\noindent
\textbf{Proof of Lemma \ref{lem:longest-suffix-prefix}}

\begin{proof}
Let $j \in [1,n]$ be an index. By definition of $C_j$ and 
$L_j$, ssince $L_j\in C_j$, we have $R_{L_j}\ge j$. Hence, by Lemma~\ref{lem:inv-cover}, the factor $w[L_j..j]$
is an inverse Lyndon word. Moreover, it is a suffix of the prefix $w[1..j]$. 

We prove that it is the \emph{longest} such suffix by contradiction. 

Assume that there exists a suffix $w[i..j]$ of $w[1..j]$ which is an inverse Lyndon word and is longer than $w[L_j..j]$. Since both factors end at position $j$, this implies that $i < L_j$. Now, since $w[i..j]$ is an inverse Lyndon word, by Lemma~\ref{lem:inv-cover} we have $R_i\ge j$. Therefore, $i\in C_j$. But this contradicts the fact that $L_j=\min C_j$, since we have found an index $i\in C_j$ such that $i<L_j$.

Hence, no longer suffix of $w[1..j]$ can be an inverse Lyndon word. Therefore, $w[L_j..j]$ is the longest one.
\end{proof}

\medskip
\noindent
\textbf{Proof of Lemma \ref{lem:L-monotone}}

\begin{proof}
By contradiction, let $j \in [1,n-1]$ such that $L_{j+1} < L_j$. 

Since $L_{j+1} \in C_{j+1}$, by the definition of $C_{j+1}$ we have $R_{L_{j+1}} \ge j+1$. In particular, $R_{L_{j+1}} \ge j$,
which implies that $L_{j+1} \in C_j$. This is impossible, however, because $L_j$ is the minimum element of $C_j$, while $L_{j+1}$ is an element of $C_j$ strictly smaller than $L_j$. 
\end{proof}

\medskip
\noindent
\textbf{Proof of Proposition \ref{prop:compute-L}}

\begin{proof}
The correctness of the computation of $R$ follows from its definition.

We prove that $L$ is computed correctly by induction on the value of $j$.

Let $j=1$. The pointer $p$ is initialized to $1$, which is the first index such that $R[p] \ge 1$, that is  the minimum of $C_1$. Thus, $L[1]$ is computed correctly.

Now assume that $j>1$. By the induction hypothesis, $L[j-1]$ has been computed correctly. By construction, when the loop begins its iteration for the value $j$, the pointer $p$ is equal to $L[j-1]$. By Lemma~\ref{lem:L-monotone}, we have $L[j] \ge L[j-1]$. At this point, we need to consider the following two cases.

If $R[p] \ge j$, then $p \in C_j$. Moreover, no index $i < p$ can belong to $C_j$. Indeed, if $i < p = L[j-1]$, then $i$ cannot belong to $C_{j-1}$, otherwise $p$ would not be the minimum of $C_{j-1}$. Hence, $R[i] < j-1$, and therefore $R[i] < j$. It follows that $i \notin C_j$. Thus, $p$ is the minimum of $C_j$, that is, $p = L[j]$.

Assume instead that $R[p] < j$. Then the \texttt{while} loop advances $p$ to the first index $q$ such that $R[q] \ge j$. Hence, by construction, $q \in C_j$. Furthermore, no index $i < q$ belongs to $C_j$: if $i < p$, then $R[i] < j$ by the previous argument, whereas if $p \le i < q$, the \texttt{while} loop explicitly checks that $R[i] < j$. It follows that $q$ is the minimum element of $C_j$, namely $q = L[j]$. 

Therefore, the algorithm correctly computes $L$.

Regarding the complexity, the array $R$ is computed by a single scan of $\lambda^{-1}$, and thus in $O(n)$ time. As for $L$, the pointer $p$ starts at $1$ and never decreases, so across the entire execution it is incremented at most $n-1$ times. It follows that $L$ is also computed in $O(n)$ time. Finally, the algorithm uses $O(n)$ space to store the arrays $R$ and $L$.
\end{proof}

\begin{example}
Let $w=\mathit{aababbaa}$ and let $x=\#\,\mathit{aababbaa}\$$. Then
\[
\lambda^{-1}_w = [2,1,3,1,4,3,2,1].
\]

For each $i=1,\dots,n$, compute
\[
R[i] = i+\lambda^{-1}_w[i]-1.
\]
Thus we obtain
\[
R[1]=2,\quad
R[2]=2,\quad
R[3]=5,\quad
R[4]=4,
\]
\[
R[5]=8,\quad
R[6]=8,\quad
R[7]=8,\quad
R[8]=8,
\]
hence
\[
R=[2,2,5,4,8,8,8,8].
\]

Now, for each $j=1,\dots,n$, compute
\[
C_j=\{\,i\in[1,j]\mid R[i]\ge j\,\}.
\]
We obtain
\[
C_1=\{1\},\qquad
C_2=\{1,2\},
\]
\[
C_3=\{3\},\qquad
C_4=\{3,4\},
\]
\[
C_5=\{3,5\},\qquad
C_6=\{5,6\},
\]
\[
C_7=\{5,6,7\},\qquad
C_8=\{5,6,7,8\}.
\]

Therefore
\[
L[j]=\min C_j
\qquad
\mbox{for each } j=1,\dots,n,
\]
and thus
\[
L[1]=1,\quad
L[2]=1,\quad
L[3]=3,\quad
L[4]=3,
\]
\[
L[5]=3,\quad
L[6]=5,\quad
L[7]=5,\quad
L[8]=5,
\]
that is,
\[
L=[1,1,3,3,3,5,5,5].
\]

By construction, for each $j$, the factor $w[L[j]..j]$ is the longest inverse Lyndon suffix of the prefix $w[1..j]$. In the present example, these factors are
\[
w[1..1]=\mathit{a},\qquad
w[1..2]=\mathit{aa},
\]
\[
w[3..3]=\mathit{b},\qquad
w[3..4]=\mathit{ba},
\]
\[
w[3..5]=\mathit{bab},\qquad
w[5..6]=\mathit{bb},
\]
\[
w[5..7]=\mathit{bba},\qquad
w[5..8]=\mathit{bbaa}.
\]

We now reconstruct $ICFL(w)$. Start with $j=8$. Since $L[8]=5$, the algorithm extracts $w[5..8]=\mathit{bbaa}$ and updates $j\gets 4$. Next, since $L[4]=3$, it extracts $w[3..4]=\mathit{ba}$ and updates $j\gets 2$. Finally, since $L[2]=1$, it extracts $w[1..2]=\mathit{aa}$ and updates $j\gets 0$, at which point the procedure stops.

Thus the factors are identified in the order $\mathit{bbaa}, \mathit{ba}, \mathit{aa}$, that is, from right to left. Writing them in the natural left-to-right order, we obtain
\[
ICFL(\mathit{aababbaa})=(\mathit{aa},\mathit{ba},\mathit{bbaa}).
\]
Hence the example exhibits explicitly the full chain
\[
\lambda^{-1}_w
\Longrightarrow
R
\Longrightarrow
(C_j)_j
\Longrightarrow
L
\Longrightarrow
ICFL(w).
\]
\end{example}

\section{Experimental Details}\label{app:experiments}

We implemented LCE-NSS for $\lambda$ and LCE-NGS for $\lambda^{-1}$ in C++17
(implementation and benchmark scripts are available online\footnote{\url{https://github.com/FLaTNNBio/inverselyndonarray}}).
The timing runs reported here were produced by the benchmark pipeline with \texttt{g++ -std=c++17 -O2}. We measured wall-clock construction time and, when relevant, report the recovery pass separately. All experiments were run on a Dell Precision 7960 Tower under Ubuntu Linux, with an Intel Xeon w5-3425 CPU and 128\,GB of ECC RAM. Correctness was checked against brute force on small and medium inputs and on crafted edge cases.

We used random words over alphabets of size $\sigma\in\{2,4,26\}$, structured synthetic families, and real texts from Pizza\&Chili and the Large Canterbury Corpus~\cite{pizza-chili,canterbury-corpus,canterbury-desc}. The profiling-enabled implementation also counted explicit character comparisons, LCE reuse hits inside the certified \proc{Smart-Lce} window, and explicit extension calls.

Table~\ref{tab:results-app} reports the random-input timings. On instances with $n\ge 5\times 10^4$, the mean ratio $\proc{LCE-NGS}/\proc{LCE-NSS}$ is $1.0060$, the median ratio is $0.9990$, and the observed range is $[0.9882,1.0939]$, which shows that the inverse construction matches the standard one closely across all alphabet sizes.

\begin{table}[tbp]
\caption{Core construction time ($\mu$s) on random words}
\label{tab:results-app}
\centering
\resizebox{\textwidth}{!}{\begin{tabular}{@{}r@{\quad}rr@{\qquad}rr@{\qquad}rr@{}}
\toprule
 & \multicolumn{2}{c}{$\sigma = 2$} & \multicolumn{2}{c}{$\sigma = 4$} & \multicolumn{2}{c}{$\sigma = 26$} \\
\cmidrule(lr){2-3}\cmidrule(lr){4-5}\cmidrule(lr){6-7}
$n$ & LCE-NSS & LCE-NGS & LCE-NSS & LCE-NGS & LCE-NSS & LCE-NGS \\
\midrule
$10^3$ & 15.4 & 15.0 & 15.8 & 15.4 & 15.6 & 16.0 \\
$5\times 10^3$ & 79.6 & 84.2 & 95.6 & 101.2 & 94.0 & 92.6 \\
$10^4$ & 171.4 & 179.0 & 205.4 & 207.0 & 194.0 & 190.6 \\
$5\!\times\!10^4$ & 948.2 & 1\,037.2 & 1\,062.0 & 1\,059.2 & 1\,006.6 & 1\,000.6 \\
$10^5$ & 1\,974.2 & 1\,973.4 & 2\,187.0 & 2\,162.4 & 2\,060.6 & 2\,053.8 \\
$5\!\times\!10^5$ & 9\,197.8 & 9\,416.2 & 10\,503.2 & 10\,543.2 & 10\,087.8 & 9\,981.4 \\
$10^6$ & 18\,060.0 & 18\,320.2 & 20\,972.4 & 20\,955.0 & 20\,121.8 & 19\,919.2 \\
$2\!\times\!10^6$ & 35\,778.4 & 36\,154.0 & 41\,893.6 & 41\,840.6 & 40\,219.0 & 39\,744.4 \\
$5\!\times\!10^6$ & 89\,178.6 & 90\,033.6 & 102\,314.0 & 104\,479.2 & 100\,569.0 & 99\,436.6 \\
\bottomrule
\end{tabular}}
\end{table}

Table~\ref{tab:real-results-app} reports the results obtained on real files. Across the five corpora, the mean ratio is $1.0019$, the median ratio is $1.0024$, and the observed range is $[0.9665,1.0325]$. The structured synthetic families show the same practical behavior. Over all structured instances with $n\ge 5\times 10^4$, the mean ratio is $0.9852$, the median ratio is $0.9981$, and the observed range is $[0.8848,1.0874]$. Thus the inverse kernel remains close to the standard one across random, structured, and real inputs.

\begin{table}[tbp]
\caption{Core construction time ($\mu$s) on real corpora}
\label{tab:real-results-app}
\centering
\begin{tabular}{@{}lrrrr@{}}
\toprule
File & $n$ & LCE-NSS & LCE-NGS & Ratio \\
\midrule
\texttt{english.txt} & 5\,000\,000 & 98\,850.6 & 98\,846.6 & 1.0000 \\
\texttt{dna.txt} & 5\,000\,000 & 104\,074.2 & 104\,328.6 & 1.0024 \\
\texttt{bible.txt} & 4\,047\,392 & 75\,372.2 & 77\,819.2 & 1.0325 \\
\texttt{e.coli} & 4\,638\,690 & 96\,704.0 & 93\,460.0 & 0.9665 \\
\texttt{world192.txt} & 2\,473\,400 & 46\,763.6 & 47\,144.6 & 1.0081 \\
\bottomrule
\end{tabular}
\end{table}

To stress the border-correction path, we additionally profiled long-border families while recording explicit character comparisons, LCE reuse hits, and explicit extension calls inside \proc{Smart-Lce}. Each instance contains a repeated prefix-suffix border occupying either $25\%$ or $40\%$ of the full length. The timing ratios on these border-heavy inputs remain close to one, mean $1.0013$, median $1.0031$, and range $[0.9680,1.0257]$. Table~\ref{tab:border-counters-app} shows that all counters remain linear in $n$. In particular, from $10^5$ to $10^6$, both explicit comparison counts and explicit extension counts grow by about $10\times$ in both families, which matches Theorem~\ref{thm:linear}. Deep borders do affect constants, but they do not trigger superlinear rescanning.

\begin{table}[tbp]
\caption{Border-heavy profiling counters. Values are averages over three runs.}
\label{tab:border-counters-app}
\centering
\scriptsize
\begin{tabular}{@{}llrrrrrr@{}}
\toprule
Family & $n$ & NSS cmp. & NGS cmp. & NSS reuse & NGS reuse & NSS ext. & NGS ext. \\
\midrule
25\% border & 100\,000    &   213\,984 &   213\,427 &  29\,347 &  29\,258 &   148\,890 &   148\,233 \\
40\% border & 100\,000    &   200\,510 &   208\,599 &  44\,719 &  35\,317 &   131\,483 &   141\,882 \\
25\% border & 500\,000    &   945\,829 & 1\,039\,872 & 294\,036 & 179\,692 &   582\,948 &   705\,698 \\
40\% border & 500\,000    &   947\,176 &   863\,646 & 292\,529 & 393\,730 &   584\,150 &   475\,202 \\
25\% border & 1\,000\,000 & 2\,134\,164 & 2\,134\,660 & 293\,205 & 293\,234 & 1\,481\,797 & 1\,483\,100 \\
40\% border & 1\,000\,000 & 2\,135\,385 & 2\,123\,841 & 292\,240 & 306\,666 & 1\,483\,741 & 1\,468\,012 \\
\bottomrule
\end{tabular}
\end{table}